\documentclass{jgcc}

\pdfoutput=1
\usepackage{lastpage}
\jgccdoi{12}{1}{2}  
\jgccheading{}{\pageref{LastPage}}{}{}%
{Feb.~4,~2020}{Mar.~20,~2020}{}

\usepackage{amsmath}
\usepackage{amssymb}
\usepackage{algorithmic}
\usepackage{booktabs}
\usepackage{multirow}
\usepackage{hyperref}

\newcommand\F{\mathbb{F}}
\newcommand\N{\mathbb{N}}
\newcommand\cZ{\mathcal{Z}}
\newcommand\cC{\mathcal{C}}
\newcommand\cP{\mathcal{P}}
\newcommand\pub{_{\text{pub}}}
\newcommand\secr{_{\text{sec}}}
\newcommand\red{^{\text{red}}}
\newcommand\lin{^{\text{lin}}}
\newcommand\GoppaGCD{{\tt GoppaGCD}}
\newcommand\I{\mathbb{I}}
\newcommand\hats{\eta}

\let\eps=\varepsilon
\let\rho=\varrho
\let\phi=\varphi

\def\tfrac#1#2{{\textstyle\frac{#1}{#2}}}
\newcommand\tsum{\textstyle\sum\limits}
\newcommand\tprod{\textstyle\prod\limits}

\newcommand\tr{^{tr}}

\newcommand\LT{\mathop{\rm LT}\nolimits}
\newcommand\encr{\mathop{\rm encr}}
\newcommand\decr{\mathop{\rm decr}}

\newcommand\Mat{\mathop{\rm Mat}\nolimits}
\newcommand\GL{\mathop{\rm GL}\nolimits}

\newcommand\wt{\mathop{\rm wt}\nolimits}

\keywords{Niederreiter cryptosystem, binary Goppa code,
side-channel analysis, fault attack}

\begin{document}

\title[A Fault Attack on the BIG-N Cryptosystem]{A Fault Attack on the
       Niederreiter Cryptosystem Using Binary Irreducible Goppa Codes}

\author{Julian Danner}
\address{Fakult\"at f\"ur Informatik und Mathematik, Universit\"at Passau,
D-94030 Passau, Germany}
\email{Julian.Danner@uni-passau.de}

\author{Martin Kreuzer}
\address{Fakult\"at f\"ur Informatik und Mathematik, Universit\"at Passau,
D-94030 Passau, Germany}
\email{Martin.Kreuzer@uni-passau.de}

\begin{abstract}
A fault injection framework for the decryption algorithm of the Niederreiter public-key
cryptosystem using binary irreducible Goppa codes and classical decoding techniques is described.
In particular, we obtain low-degree polynomial equations in parts of the secret key.
For the resulting system of polynomial equations, we present an efficient solving strategy
and show how to extend certain solutions to alternative secret keys.
We also provide estimates for the expected number of required fault injections,
apply the framework to state-of-the-art security levels,
and propose countermeasures against this type of fault attack.
\end{abstract}

\maketitle

%%%%%%%%%%%%%%%%%%%%%%%%%%%%%%%%%%%%%%
%
%  Section 1: Introduction
%
%%%%%%%%%%%%%%%%%%%%%%%%%%%%%%%%%%%%%%

\section{Introduction}

Many established public-key cryptosystems rely on the hardness of the
factorization problem or the discrete logarithm problem. However, their
long-term security is not guaranteed, because they can be broken in
polynomial time using sufficiently large quantum-computers~\cite{shor1999polynomial}.
This motivates the need for post-quantum cryptosystems.
One of the oldest public-key cryptosystems is the McEliece
cryptosystem~\cite{mceliece1978public} which was designed to be used by NASA.
However, partly due to its large public-key sizes, it was never standardized.
The security of the McEliece cryptosystem relies on the hardness of the
decoding problem for random linear codes~\cite{overbeck2009code}.

The Niederreiter cryptosystem is a variant of the McEliece cryptosystem
which offers some improvements as to the costs of encryption and decryption
and requires smaller public-key sizes~\cite{niederreiter1986knapsack}.
Particularly promising variants are based on the problem of decoding
binary Goppa codes~\cite{dinh2011mceliece}.

For some special subclasses of Goppa codes, namely quasi-dyadic and
quasi-cyclic Goppa codes, there exist successful algebraic attacks which
take advantage of the particular structure of the code~\cite{faugere2010algebraic}.
But apart from these subclasses, binary (irreducible) Goppa codes still appear
to resist structural attacks.
The best known algorithms for decoding arbitrary linear codes use
information set decoding, and their running time is still exponential~\cite{may2011decoding}.

Structural attacks on the Niederreiter and McEliece cryptosystems have been widely
researched. However, their side-channel analysis is not quite as advanced.
In general, any additional source of information that can be derived from a specific hardware
or software implementation of a cryptosystem, or even its execution, may be
considered a side-channel. Clearly, side-channel attacks have to be measured by
their practical feasibility. Passive side-channel analysis is frequently based
on power-analysis of hardware devices or timings of the execution
of certain parts of the encryption or decryption algorithms.

In this paper we consider a type of active side-channel attacks
called fault attacks. In particular, we assume that we are able to inject
an error into the usual flow of the algorithms of the cryptosystem by
corrupting the contents of specific memory cells at a particular moment.
Common methods of achieving such a fault injection are manipulations
of the power supply and the usage of pulsed laser or ion beams.
The usability of a fault attack depends chiefly on the fault model
which describes the required physical capabilities of the attacker.

For the Niederreiter and McEliece cryptosystems, there exist successful passive
side-channel attacks~\cite{avanzi2011side,rossi2017side,strenzke2012fast,strenzke2008side}
which exploit traditional side-channels, such as timing and power consumption attacks.
For active side-channel attacks such as fault attacks, much less seems
to be known. We found only the article~\cite{cayrel2010mceliece} which analyses
the sensitivity of these cryptosystems to fault injections in the encryption
and key generation algorithms.

However, the most natural target for a fault attack at a public-key cryptosystem is
the decryption algorithm, because it uses the secret key whose knowledge
would allow us to completely break the system. In fact, it suffices to find
an {\it alternative key}, i.e., a key which also allows us to decrypt ciphertexts.
The current paper is
a first attempt at constructing such fault attacks. Our target class of
cryptosystems are Niederreiter cryptosystems based on binary irreducible Goppa codes.
We call them briefly {\it \mbox{BIG-N} cryptosystems}. The attack is based on the following
fault model:
\begin{enumerate}
\item[(1)] The decryption algorithm follows the standard pattern:
first the error locator polynomial is computed, and then  it is evaluated
at the support elements to deduce the plaintext.

\item[(2)] The decryption algorithm not only reconstructs plaintext units
of weight~$t$, the designed distance, but also {\it illegal} plaintexts~$p$, i.e.,
plaintexts of weight $1\le \wt(p) < t$.

\item[(3)] After the error-locator polynomial has been computed,
we are able to inject a uniformly random fault into a chosen coefficient
of this polynomial, i.e., we are able to replace it by a random
bit-tuple.
\end{enumerate}
Moreover, we assume that we are able to repeat these injections
hundreds or even several thousand times. The resulting \mbox{BIG-N} fault attack
breaks all state-of-the-art security levels~\cite{bernstein2008attacking},
even ``long term'' secure ones, within minutes. Consequently, we suggest
explicit countermeasures for hardware- or software-implementations
of \mbox{BIG-N} cryptosystems.

{\it Feasibility and Applicability.} Let us briefly discuss the practical applicability
and relevance of the new fault attack. As we shall see in Section~\ref{sec:timings},
for carrying out the actual fault injections, we need to hit the register
holding a certain 10-13 bit wide coefficient at a specified point in time.
Using modern equipment, this is a mild requirement (see for instance~\cite{Joye2012}
and~\cite{Breier2019}).

Thus it seems more pertinent to discuss the vulnerability
of current cryptographic \mbox{BIG-N} schemes to the proposed attacks.
Looking through the NIST Post-Quantum Standardization candidates, there are
three original submissions based on Goppa codes: ``BIG Quake'' which is
vulnerable to our attack, ``Classic McEliece'', and ``NTS-KEM''. The latter two
are actually key exchange systems where the key cannot be chosen freely,
and for the random vector a hash value is transmitted. This is clearly a setting
in which the current attack does not apply directly. However, the reference
implementation~\cite{wang2018fpga} for the \mbox{BIG-N} part of the schemes is
not protected against the attack. Furthermore, the other published
reference implementations~\cite{heyse2013code} and~\cite{hu2018} of the \mbox{BIG-N}
cryptosystem use constant weight encoders and decoders whose standard implementation
is also vulnerable. Thus, although it may not be very difficult to defend
against, apparently all current implementations and applications of the
\mbox{BIG-N} cryptosystem are vulnerable to the \mbox{BIG-N} fault attack.

{\it Contents.} Let us describe the structure of the paper in more detail.
After recalling some basic facts about binary Goppa codes and
\mbox{BIG-N} cryptosystems in Sections~2 and~3, we present a \mbox{BIG-N} fault
attack framework in Section~4. In particular, we analyze the
assumptions underlying the attack carefully and propose countermeasures.
Then we introduce the general framework for the fault attack:
we assume that we are able to replace the error locator polynomial
$\sigma_e(x)$ by an erroneous one of the form $\tilde{\sigma}_e(x)=
\eps x^d + \sigma_e(x)$ where $\eps\in \F_{2^m}$ is distributed
uniformly at random and~$d$ is the chosen degree under attack.

In Section~5 we analyze the resulting equations for the
components of the support vector $\alpha=(\alpha_1,\dots,\alpha_n)$
of the binary Goppa code in two particular cases:
we attack the constant and the quadratic coefficient of~$\sigma_e(x)$.
From a successful constant injection we derive a linear equation
for the components of~$\alpha$, and from a successful quadratic
injection we get a linear or a quadratic equation. However, this
typically requires a sequence of injections until we succeed in
obtaining an erroneously deciphered word of weight two.

The next steps are taken in Section~6 where we combine the
acquired linear and quadratic equations into a fault equation system
and then carry out the actual \mbox{BIG-N} fault attack in three steps:
Firstly, we solve the fault equation system and get a set of support
candidates. Secondly, using the fact that it suffices to find the
support and the Goppa polynomial of a larger binary Goppa code
containing the publicly known one, we determine a support candidate which
can be extended. Finally, we use this alternative secret key to
break the given \mbox{BIG-N} cryptosystem.

In the final section we offer some experiments and timings
for the \mbox{BIG-N} fault attack. In particular, we provide estimates
for the average numbers of constant and quadratic fault injections
needed to succeed. Moreover, we collect the timings for
breaking various security levels, ranging from one minute for
``short term'' 60-bit security to about 25 minutes for
``long term'' 266-bit security.

%%%%%%%%%%%%%%%%%%%%%%%%%%%%%%%%%%%%%%%%%%%%%%%%%%%
%
% Section 2: Binary Goppa Codes
%
%%%%%%%%%%%%%%%%%%%%%%%%%%%%%%%%%%%%%%%%%%%%%%%%%%%

\section{Binary Goppa Codes}

For starters, let us recall the definition of a binary Goppa code.
By $\F_{2^m}$ we denote the finite field having $2^m$ elements.
We also fix the following notation:
given a tuple $c=(c_1,\dots,c_n)\in\F_2^n$, we let
$\I_c:=\{i\in\{1,\dots,n\}\mid c_i=1\}$.

\begin{defi}
Let $m,t,n\in\N_+$ such that $mt<n\leq 2^m$.
\begin{enumerate}
\item[(a)] A tuple $\alpha=(\alpha_1,\dots,\alpha_n)\in\F_{2^m}^n$ such that
$\alpha_i\neq\alpha_j$ for $i\neq j$ will be called a {\bf support tuple}.

\item[(b)] A polynomial $g\in\F_{2^m}[x]$ with $\deg(g)=t$ and
$g(\alpha_i)\neq0$ for $i\in\{1,\dots,n\}$ is called a {\bf Goppa polynomial}
for the support tuple~$\alpha$.

\item[(c)] The {\bf (binary) Goppa code} with the generating pair $(\alpha,g)$
is given by
\[
\Gamma(\alpha,g)=\{c\in\F_2^n\mid \tsum_{i\in \I_c} (x-\alpha_i)^{-1} = 0 \text{ in }
\F_{2^m}[x]/\langle g\rangle \}
\]
In particular, if~$g$ is an irreducible polynomial, then $\Gamma(\alpha,g)$ is
called an {\bf irreducible Goppa code}.
\end{enumerate}
\end{defi}

\begin{rem}\label{rmk:suppcodegoppapoly}
For a Goppa code $C=\Gamma(\alpha,g)$ and for $c=(c_1,\dots,c_n)\in\F_2^n$, we have
\begin{align*}
  c\in C &\quad\Leftrightarrow\quad \tsum_{i\in \I_c}(x-\alpha_i)^{-1} = 0
                                             \tag{in $\F_{2^m}[x]/\langle g\rangle $}\\
         &\quad\Leftrightarrow\quad g \;\mid\; \tsum_{i\in \I_c} \tprod_{j\in \I_c\setminus\{i\}}
         (x-\alpha_j)          \tag{in $\F_{2^m}[x]$}
\end{align*}
\end{rem}

\begin{rem}\label{rmk:paritycheckmat}
Given a Goppa code $C=\Gamma(\alpha,g)$ it is well-known that a parity-check matrix
$H\in\Mat_{mt\times n}(\F_2)$ of~$C$ can be obtained from the matrix
\[
H' \;=\;\begin{pmatrix}
\beta_1 & \cdots & \beta_n\\
\alpha_1 \beta_1 & \cdots & \alpha_n \beta_n\\
\vdots & & \vdots\\
\alpha_1^{t-1} \beta_1 & \cdots & \alpha_n^{t-1} \beta_n
\end{pmatrix}  \;\in\;  \Mat_{t,n}(\F_{2^m})
\]
where $\beta_i=g(\alpha_i)^{-1}$ for $i=1,\dots,n$,
by replacing each entry of~$H'$ by a column of~$m$ bits
that arise by fixing an $\F_2$-basis of~$\F_{2^m}$.
\end{rem}

From here on, let $C=\Gamma(\alpha,g)$ be a binary Goppa code
with parameters $m,t,n\in\N_+$ as described above,
and let $H\in\Mat_{mt,n}(\F_2)$ be a parity-check matrix of~$C$.
For $\tilde{c}\in\F_2^n$, we call $s_{\tilde c} = \tilde{c} H\tr
\in\F_2^{mt}$ the {\bf syndrome} of~$\tilde{c}$ with respect to~$H$.
Then we have $\tilde{c}\in C$ if and only if $\tilde{c} H\tr=0$.

It is known that $\dim C\ge n-mt$ and that the minimal distance of~$C$ satisfies
$d_{\min}(C)\geq t$. If~$C$ is irreducible, we even have $C=\Gamma(\alpha,g^2)$
and $d_{\min}(C)\geq 2t$. Hence, in general, up to~$\frac{t}{2}$ errors can be corrected,
and if $C$ is irreducible, even up to~$t$ errors can be corrected uniquely.

Since we are going to use it extensively, let us briefly recall the classical syndrome
decoding method for binary Goppa codes.

\begin{rem}{\bf (Sydrome Decoding for Goppa Codes)}\\
Consider a received word $\tilde c=c+e\in\F_2^n$ with $c\in C$ and $e\in\F_2^n$.
Then we define the {\bf error-locator polynomial} by
\[
\sigma_e(x) \;=\; \tprod_{i\in \I_e} \, (x-\alpha_i) \;\in\; \F_{2^m}[x]
\]
and the {\bf syndrome polynomial} of $\tilde{c} = (\tilde{c}_1,\dots,\tilde{c}_n)\in\F_2^n$ by
\[
s_{\tilde c}(x) \;=\; \tsum_{i=1}^n \tfrac{\tilde c_i}{g(\alpha_i)}
\tfrac{g(x)-g(\alpha_i)}{x-\alpha_i}   \;\in\;  \F_{2^m}[x]
\]
Then we have $s_e(x)\equiv s_{\tilde c}(x)\mod g(x)$, and we obtain the {\bf key equation}
\[
\sigma_e(x)\cdot s_{\tilde c}(x) \;\equiv\; \sigma'_e(x)\mod g(x)
\]
Let us write $g(x) = g_tx^t+\dots+g_1x+g_0$ with $g_0,\dots,g_t\in \F_{2^m}$.
Using the same $\F_2$-basis of~$\F_{2^m}$ as in Remark~\ref{rmk:paritycheckmat},
we combine sequences of~$m$ entries in the syndrome $s_{\tilde c}
=\tilde{c} H\tr = e H\tr \in\F_2^{mt}$ and get $\hat{s}\in\F_{2^m}^t$.
Now the coefficients of the syndrome polynomial $s_{\tilde c}(x)$ can be computed
by a simple multiplication of $\hat{s}\in\F_{2^m}^t$ with the matrix
\[
S_g \;=\; \begin{pmatrix}
    g_t    & g_{t-1} & \cdots  & g_1 \\
           & \ddots  & \ddots  & \vdots\\
           &         & g_t     & g_{t-1}\\
           &         &         & g_t
\end{pmatrix}  \;\in\; \GL_t(\F_{2^m})
\]
Hence it is sufficient to solve the key equation for the error-locator polynomial $\sigma_e(x)$.
Then the zeros of $\sigma_e(x)$ determine the error vector $e\in\F_2^n$ via the
oberservation that, for $i\in\{1,\dots,n\}$, we have $e_i=1$ if and only if $\sigma_e(\alpha_i)=0$.
Finally, we decode~$\tilde c$ to $c=\tilde c+e$.
\end{rem}

The main task in this method is to solve the key equation. This can be done in several ways,
for instance explicitly using the Sugiyama-Algorithm~\cite{sugiyama1975method},
or implicitly using the Berlekamp-Massey Algorithm~\cite{berlekamp1966non,massey1969shift}.
As shown in~\cite{Dornstetter1987equiv}, one may consider these two algorithms
as essentially equivalent.
For a binary irreducible Goppa code~$C$, up to~$t$ errors can be corrected
using the Patterson Algorithm~\cite{patterson1975algebraic}
which also uses the key equation to obtain the error-locator polynomial.
Independently of the chosen method, decoding consists of the following
three basic steps.

\begin{algo}{\bf (The Syndrome Decoding Algorithm)}\label{alg:syndromedecoding}\\
Let $C=\Gamma(\alpha,g)\subseteq\F_2^n$ be a binary Goppa code, where
$\alpha=(\alpha_1,\dots,\alpha_n)\in\F_{2^m}^n$ and $g\in\F_{2^m}[x]$,
and let $t=\deg(g)$.
Let $H\in\Mat_{mt\times n}(\F_2)$ be a parity check matrix of~$C$ obtained
as in Remark~\ref{rmk:paritycheckmat}, and let $s_{\tilde c}=\tilde{c} H\tr \ne 0$
be the syndrome of a given received word $\tilde{c}\in \F^n \setminus C$.
We assume that $\tilde{c}$ is of the form $\tilde{c}=c+e$ with $c\in C$ and
$e\in \F_2^n$ such that $\#\I_e \le \frac{t}{2}$ (or $\#\I_e\le t$, if~$C$
is irreducible). Then we compute~$e\in \F_2^n$ from $s_{\tilde c}$
using the following steps.
\begin{enumerate}
\item[(1)] Compute the syndrome polynomial $s_{\tilde c}(x)\in\F_{2^m}[x]$
from $s_{\tilde c}$.

\item[(2)] Compute the error-locator polynomial $\sigma_e(x)$ by solving the
key equation (explicitly or implicitly), e.g., via one of the cited algorithms.

\item[(3)] For $i=1,\dots,n$, let $e_i=1$ if $\sigma_e(\alpha_i)=0$
and $e_i=0$ otherwise. Return $e=(e_1,\dots,e_n)\in\F_2^n$ and stop.
\end{enumerate}
\end{algo}

%%%%%%%%%%%%%%%%%%%%%%%%%%%%%%%%%%%%%%
%
%  Section 3: BIG-N Cryptosystems
%
%%%%%%%%%%%%%%%%%%%%%%%%%%%%%%%%%%%%%%

\section{BIG-N Cryptosystems}\label{sec:bignc}

In this section we introduce Niederreiter cryptosystems using binary irreducible Goppa codes
(see~\cite{niederreiter1986knapsack}) and recall some of their basic properties.
Subsequently, we write $W_{n,t}$ for the set of all elements $v\in \F_2^n$ of {\bf Hamming
weight}~$t$, i.e., such that $\wt(v)=\#\I_v =t$.

\begin{defi}{\bf (BIG-N Cryptosystems)}\label{defi:nrirred}\\
A {\bf Niederreiter cryptosystem using a binary irreducible Goppa code~$C$ (\mbox{BIG-N} cryptosystem)}
with parameters $m,t,n\in\N_+$ such that $mt<n\leq 2^m$ is represented by
a tuple $(K\secr,K\pub,\cP,\cC,\encr,\decr)$ consisting of the following parts.
\begin{enumerate}
\item[(a)] The tuple $K\secr=(S,H,P,\alpha,g)$, where $H\in\Mat_{mt\times n}(\F_2)$ is a
parity check matrix of the irreducible Goppa code $C=\Gamma(\alpha,g)$ with $\deg(g)=t$
and dimension $n-mt$, where $P\in\Mat_n(\F_2)$ is
a permutation matrix, and where $S\in\GL_{mt}(\F_2)$, is called the {\bf secret key}.

\item[(b)] The tuple $K\pub=(m,t,n,H\pub)$, where $m,t,n$ are the parameters of~$C$,
and where $H\pub=S\cdot H\cdot P\in\Mat_{mt,n}(\F_2)$,
is called the {\bf public key}.
The matrix $H\pub$ is also called the {\bf public parity check matrix}.

\item[(c)] The set $\cP= W_{n,t}$ is called the {\bf plaintext space},
and an element $p\in\cP$ is called a {\bf plaintext unit}.

\item[(d)] The set $\cC=\{pH\pub\tr\mid p\in\cP\}$ is called the {\bf ciphertext space},
and an element $c\in\cC$ is called a {\bf ciphertext unit}.

\item[(e)] The map $\encr\colon\cP \longrightarrow \cC$ given by $\encr(p)=pH\pub\tr$
is called the {\bf encryption map}.

\item[(f)] The map $\decr\colon\cC \longrightarrow \cP$ given by
$\decr(c)=\gamma(c\cdot(S\tr)^{-1})\cdot(P\tr)^{-1}$
is called the {\bf decryption map}. Here the map $\gamma:\; \F_2^{mt}\longrightarrow\F_2^n$
satisfies $\gamma(e\,H\tr)=e$ for all $e\in W_{n,t}$ and is a syndrome decoding algorithm
which corrects up to~$t$ errors.
\end{enumerate}
\end{defi}

In particular, notice that we have $\decr(\encr(p))=p$ for all $p\in\cP$.
Indeed, for a plaintext unit $p\in\cP$ and its ciphertext $c=\encr(p)=pH\pub\tr$,
we have $\wt(pP\tr) = \wt(p) = t$, and therefore $\gamma(c\,(S\tr)^{-1}) =
\gamma((p\,P\tr)H\tr)=p\,P\tr$. Hence we get $\decr(c)=p$, as claimed.

A \mbox{BIG-N} cryptosystem is a public-key cryptosystem, i.e., the encryption map $\encr$
can be computed solely using the public key $K\pub$, but the application of
the decryption map $\decr$ requires the knowledge of the secret key $K\secr$.
The cryptosystem $(K\secr,K\pub,\cP,\cC,\encr,\decr)$ is considered broken
if one can efficiently compute a fast {\bf alternative decryption map}
$\decr':\; \cC \longrightarrow \cP$ which satisfies $\decr'(\encr(p))=p$ for all $p\in\cP$.
If, for an alternative decryption map~$\decr'$, there exists a secret key $K\secr'$
such that $(K\secr',K\pub,\cP,\cC,\encr,\decr')$ is a \mbox{BIG-N} cryptosystem,
then $K\secr'$ is called an {\bf alternative secret key}.

In the following we let $m,t,n\in\N_+$ such that $mt<n\leq 2^m$, we let
$C=\Gamma(\alpha,g)$ be an irreducible Goppa code with these parameters, and we
let $(K\pub,K\secr,\cP,\cC,\encr,\decr)$ be a \mbox{BIG-N} cryptosystem using~$C$.
Moreover, we fix an $\F_2$-basis of~$\F_{2^m}$ and use it to convert
elements of~$\F_{2^m}$ to $m$-tuples of bits and vice versa.

\begin{rem}{\bf (Simplified \mbox{BIG-N} Cryptosystems)}\label{rmk:simplify}\\
In the above setting, let $H\in\Mat_{mt,n}(\F_2)$ be the parity check matrix
for~$C$ obtained as in Remark~\ref{rmk:paritycheckmat}, and let
$P\in\Mat_n(\F_2)$ be the permutation matrix contained in the secret key.
Then we define $\tilde{\alpha} := \alpha\cdot P\in\F_{2^m}^n$ and
$\tilde{C}=\Gamma(\tilde\alpha,g)$. Clearly, also~$\tilde{C}$ is a binary
irreducible Goppa code, and the matrix $\tilde{H} :=H\cdot P$ is a parity check matrix
for~$\tilde{C}$ which has the shape described in Remark~\ref{rmk:paritycheckmat}.

Since the public parity check matrix satisfies $H\pub=S\cdot \tilde{H}$,
an alternative secret key for the original cryptosystem is given by
$(S,\tilde H,I_n,\tilde\alpha,g)$, where $I_n$ is the identity matrix of size~$n$.
Consequently, we may use the Goppa code~$\tilde{C}$ instead of~$C$ and
get rid of the permutation matrix~$P$ as part of the secret key.

From here on we use secret keys of the form $K\secr=(S,H,\alpha,g)$
and simplify all \mbox{BIG-N} cryptosystems accordingly.
Notice that the code defined by the public parity check matrix~$H\pub$
of the simplified cryptosystem is now equal to~$C$
and continues to be publicly known.
\end{rem}

In view of the preceding section, we know that a \mbox{BIG-N} cryptosystem is broken
if an efficient $t$-error correcting algorithm is found for the Goppa code~$C$.
The next algorithm shows that this requirement can be weakened even further.

\begin{algo}{\bf (An Alternative Decryption Algorithm)}\label{alg:altkey}\\
Consider a \mbox{BIG-N} cryptosystem $(K\secr,K\pub,\cP,\cC,\encr,\decr)$ using the
Goppa code $C=\Gamma(\alpha,g)$ and having the public key $K\pub=(m,t,n,H\pub)$.
Assume that there exist a support tuple $\tilde{\alpha}\in\F_{2^m}^n$
and a Goppa polynomial $\tilde{g}\in\F_{2^m}[x]$ for~$\tilde{\alpha}$
such that $\tilde{t}:=\deg(\tilde{g})\ge 2t$ and
$C \subseteq \tilde{C} := \Gamma(\tilde{\alpha},\tilde{g})$.
Given a ciphertext unit $c\in\cC$, consider the following sequence of instructions.
\begin{enumerate}
\item[(1)] Compute a parity check matrix $\tilde{H} \in\Mat_{m\tilde{t},n}(\F_2)$
for~$\tilde{C}$ as in Remark~\ref{rmk:paritycheckmat}.

\item[(2)] Compute the matrix $\tilde{S}\in\Mat_{m\tilde{t},mt}(\F_2)$ such that
$\tilde{S}\cdot H\pub=\tilde{H}$.

\item[(3)] Compute $\tilde{c} = c\, \tilde{S}\tr$ and apply the syndrome decoding
algorithm of~$\tilde{C}$ to~$\tilde{c}$. Return the resulting tuple~$p$.
\end{enumerate}
This is an algorithm which computes a plaintext unit $p\in\cP$
such that we have $\encr(p)=c$.
\end{algo}

\begin{proof}
First notice that $C \subseteq \tilde{C}^\perp$ implies
$\tilde{C}^\perp\subseteq C^\perp$. Since the rows of~$\tilde{H}$ generate~$\tilde{C}^\perp$
and the rows of~$H\pub$ generate $C^\perp$, it follows that there
exists a matrix $\tilde{S}\in\Mat_{m\tilde{t},mt}(\F_2)$ such that
$\tilde{S}\cdot H\pub=\tilde{H}$, and this matrix can computed in step~(2).
Furthermore, the input for Algorithm~\ref{alg:syndromedecoding} in step~(3) is correct,
because $c=\encr(p)$ for some $p\in\cP$, and therefore $\tilde{c} = c\, \tilde{S}\tr
=p\, H\pub\tr \tilde{S}\tr = p\,\tilde H\tr$ is a syndrome of weight $\wt(p)=t$
with respect to~$\tilde{H}$. Consequently, the algorithm can be executed.

Its finiteness is clear. Its correctness follows from the fact that
Algorithm~\ref{alg:syndromedecoding} corrects up to $\tfrac{\deg(\tilde{g})}{2} =
\tfrac{\tilde{t}}{2}\ge t = \wt(p)$ errors, and hence determines $p\in\cP$ correctly.
\end{proof}

In view of this algorithm, it is clear that a \mbox{BIG-N} cryptosystem using the code
$C=\Gamma(\alpha,g)$ is broken if a generating pair $(\tilde{\alpha},\tilde{g})$
for a binary Goppa code is found such that $\deg(\tilde{g}) \ge 2t$ and such that
$C\subseteq \tilde{C}=\Gamma(\tilde{\alpha},\tilde{g})$. Therefore
such a pair $(\tilde{\alpha},\tilde{g})$ will be called an {\bf alternative secret pair}.

%%%%%%%%%%%%%%%%%%%%%%%%%%%%%%%%%%%%%%%%%%%%%%%%%%%
%
%  Section 4: The BIG-N Fault Injection Framework
%
%%%%%%%%%%%%%%%%%%%%%%%%%%%%%%%%%%%%%%%%%%%%%%%%%%%

\section{The BIG-N Fault Injection Framework}\label{sec:finjfw}

In this section we let $m,t,n\in\N_+$ such that $mt<n\leq 2^m$, we let
$C=\Gamma(\alpha,g)$ be a binary irreducible Goppa code with these parameters, and we
let $(K\pub,K\secr,\cP,\cC,\encr,\decr)$ be a \mbox{BIG-N} cryptosystem using~$C$.
Moreover, we write $K\pub=(m,t,n,H\pub)$ for the public key and
$K\secr=(S,H,\alpha,g)$ for the secret key. Using Remark~\ref{rmk:simplify},
we assume that $H\pub=S\cdot H$. Recall that we let
$\I_p = \{i\in\{1,\dots,n\}\mid p_i=1\}$ for $p=(p_1,\dots,p_n)\in\F_2^n$.

In order to mount the proposed fault attack, we require that the implementation
of the decryption map $\decr$ satisfies three assumptions. They are motivated
by the following {\it usual} implementation which is based on the classical
syndrome decoding method of Algorithm~\ref{alg:syndromedecoding}.

\begin{algo}{\bf (Implementing the Decryption Map $\decr$)}\label{alg:decrypt}

{\bf Input:} a ciphertext unit $c=\encr(p)\in\cC$, the secret key $K\secr=(S,H,\alpha,g)$

{\bf Output:} a plaintext unit $p\in\cP$

\begin{algorithmic}[1]
  \STATE Compute the syndrome $s=c(S\tr)^{-1}\in\F_2^{mt}$ with respect to~$H$.
  \STATE Compute the syndrome polynomial $s_{p}(x)$.
  \STATE Compute the error-locator polynomial $\sigma_p(x)$.
  \STATE Determine $p=(p_1,\dots,p_n)\in\cP$ by setting $p_i=1$ if
         $\sigma_p(\alpha_i)=0$ and $p_i=0$ otherwise.
  \RETURN{$p$}
\end{algorithmic}
\end{algo}

In view of this algorithm, the following assumption seems natural.

\begin{asm}\label{ass:1}
The implementation of the decryption map $\decr$ makes use of the error-locator polynomial
in such a way that it is first computed explicitly, then evaluated at the support elements,
and finally the resulting plaintext unit is returned.
\end{asm}

Next we let $p\in\F_2^n$ with $0 < \wt(p)<t$ and $c=p\, H\pub\tr$.
Since Algorithm~\ref{alg:decrypt} is based on the Syndrome Decoding
Algorithm~\ref{alg:syndromedecoding}, and since syndrome decoding
corrects {\it up to}~$t$ errors, applying this algorithm to~$c$
will correctly return~$p$. Hence we make the following assumption.

\begin{asm}\label{ass:2}\label{ass:assumptions}
Let $p\in\F_2^n$ with $0 < \wt(p)\le t$, and let $c=pH\pub\tr$.
If we apply the decryption map $\decr$ to~$c$, it returns~$p$.
\end{asm}

In order to be able to inject a fault into the decryption process, we
require one final assumption.

\begin{asm}\label{ass:3}
Let $d\in\N$ with $d<t$, let $p\in\F_2^n$ with $0 < \wt(p) \le t$, and let
$c=p\,H\pub\tr$. After the error-locator polynomial $\sigma_p(x)$
has been computed during the decryption process of~$c$,
we assume that we can inject a uniformly random fault into the $d$-th
coefficient of~$\sigma_p(x)$. In other words, we assume that we
may replace $\sigma_p(x)$ by a polynomial $\tilde{\sigma}_p(x)=\eps x^d+\sigma_p(x)$
where $\eps \in\F_{2^m}$ is chosen uniformly at random.
\end{asm}

As a testimony to the applicability of these assumptions, consider the hardware
implementation described in~\cite{wang2018fpga}.
It follows the steps of Algorithm~\ref{alg:decrypt}, and after $\sigma_p(x)$ has
been computed, this polynomial is transferred to the evaluation module
by sending $m(t+1)$ bits, where each $m$-bit subtuple
represents one of coefficients of~$\sigma_p(x)$.
This transfer process is suited for an injection of a (uniformly distributed) random
error $\eps\in\F_{2^m}$ into the $d$-th coefficient of $\sigma_p(x)$.
From a hardware point of view, it corresponds to randomly changing the state of~$m$
chosen consecutive bits.

Every implementation of the decryption map that satisfies these assumptions is vulnerable
to the following fault injection framework.

\begin{algo}{\bf (The \mbox{BIG-N} Fault Injection Framework)}\label{alg:faframe}\\
For a \mbox{BIG-N} cryptosystem $(K\secr,K\pub,\cC,\cP,\encr,\decr)$ as above, assume
that the implementation of the decryption map~$\decr$
satisfies Assumptions~\ref{ass:1}, \ref{ass:2}, and~\ref{ass:3}.
Choose a word $p\in\F_2^n$ with $0 < \wt(p) \le t$ and a number $d\in\N$ with $d<t$.
Then consider the following sequence of instructions.

\begin{enumerate}
\item[(1)] Compute $c=pH\pub\tr\in\F_2^{mt}$.

\item[(2)] Start the decryption algorithm $\decr$ with input $c\in\F_2^{mt}$
and inject a uniformly random fault $\eps\in\F_{2^m}$ in the $d$-th coefficient
of $\sigma_p(x)$ such that $\tilde\sigma_p(x)=\eps x^d+\sigma_p(x)$ is evaluated
instead of $\sigma_p(x)$.

\item[(3)] Return $\tilde p\in\F_2^n$, the output of the faulty decryption of step~(2).
\end{enumerate}
This is an algorithm which returns a tuple $\tilde{p} = (\tilde{p}_1,\dots,\tilde{p}_n)
\in\F_2^n$ such that, for $i\in \{1,\dots,n\}$, we have $\tilde{p}_i=1$
if and only if $\eps \alpha_i^d + \prod_{j\in \I_p}(\alpha_i-\alpha_j) =0$.

Consequently, every component $\tilde{p}_i=1$ yields a polynomial equation in
$\F_{2^m}[x_0,\dots,x_n]$ which is satisfied for $(\eps,\alpha_1,\dots,\alpha_n)$
\end{algo}

\begin{proof}
By Assumption~\ref{ass:2}, the decryption algorithm is correct for all
syndromes $c=p\, H\pub\tr$ with $p\in\F_2^n$ and $0 < \wt(p) \le t$.
In combination with Assumption~\ref{ass:1}, this means that in the course of the
decryption algorithm, the error-locator polynomial $\sigma_p(x)$ is computed
correctly. By Assumption~\ref{ass:1}, the output $\tilde p=(\tilde{p}_1,\dots,\tilde{p}_n)
\in\F_2^n$ satisfies $\tilde p_i=1$ if and only if $\tilde\sigma_p(\alpha_i)=0$.
Since we have $\tilde{\sigma}_p(x) = \eps x^d + \sigma_p(x)$ and
$\sigma_p(x)=\prod_{i\in \I_p}(x-\alpha_i)$, the claim follows.
\end{proof}

In the setting of this framework, we call the triple $(p,d,\tilde{p})$
a {\bf \mbox{BIG-N} fault injection} in degree~$d$. We also say that this injection
{\bf uses the fault}~$\eps$. The possibility to perform \mbox{BIG-N} fault injections
can be prevented as follows.

\begin{rem}{\bf (Countermeasures)}\label{rmk:countermeasures}\\
Let $(p,d,\tilde{p})$ be a \mbox{BIG-N} fault injection.
\begin{enumerate}
\item[(a)] The output of the decryption map of a \mbox{BIG-N} cryptosystem
is an $n$-bit tuple of weight~$t$. In general, for a \mbox{BIG-N} fault injection
$(p,d,\tilde{p})$, the output~$\tilde{p}$ will have weight $\le t$.
Thus checking the weight of~$\tilde{p}$ discovers most fault injections.

\item[(b)] A further way to detect fault injections is to
re-encrypt the output~$\tilde{p}$. If$\tilde{p}\ne p$, we will
get $\tilde{p}\,H\pub\tr \ne c$.
\end{enumerate}
\end{rem}

The next proposition collects some observations on \mbox{BIG-N} fault injections.

\begin{prop}\label{prop:faulteq2}
Let $(p,d,\tilde p)$ be a \mbox{BIG-N} fault injection which uses the fault $\eps\in\F_{2^m}$,
and let $\tilde{\sigma}_p(x)=\eps x^d+\sigma_p(x)$.
Then the following claims hold.
\begin{enumerate}
\item[(a)] If $\wt(\tilde p) > \wt(p)$ then $\eps\ne 0$.

\item[(b)] We have either $p=\tilde{p}$ or $\# (\I_p\cap \I_{\tilde p}) = 1$.

\item[(c)] If $\eps\ne 0$ then $\I_p\cap \I_{\tilde p}=\{i\}$ is equivalent to
$i\in \I_p$, $\alpha_i=0$, and $d>0$.
\end{enumerate}
\end{prop}

\begin{proof}
First we prove~(a). Assuming that $\wt(\tilde{p}) > \wt(p)$, we have to show
$\eps\ne 0$. Considering the way in which~$\tilde{p}$ is determined by the zeros
of~$\tilde\sigma_p(x)$, we see that $\deg(\tilde{\sigma}_p(x)) \ge \wt(\tilde{p})$.
Since $\wt(p)=\deg(\sigma_p(x))$, we get $\deg(\tilde{\sigma}_p(x)) > \deg(\sigma_p(x))$,
and hence $\eps\ne 0$.

To show~(b), we consider two cases. In the case $\eps=0$,
we clearly have $\I_p = \I_{\tilde p}$ and thus $p=\tilde p$.
It remains to examine the case $\eps\ne 0$.
For a contradiction, assume that $\# (\I_p\cap \I_{\tilde p}) \ge 2$.
Let $j_1,j_2\in \I_p\cap \I_{\tilde p}$ with $j_1\ne j_2$.
Then $\eps x^d=\tilde{\sigma}_p(x)-\sigma_p(x)$ has the distinct zeros $\alpha_{j_1}$
and $\alpha_{j_2}$. This is a contradiction to the fact that~$\eps x^d$ has no
two distinct zeros.
Consequently, we get $\# (\I_p\cap \I_{\tilde p})\le 1$.

Finally, we prove~(c). Let $\eps\ne 0$. To show the implication ``$\Rightarrow$'',
we note that $\alpha_i$ is a zero of $\tilde{\sigma}_p(x)$ and of $\sigma_p(x)$.
Hence $\alpha_i$ is a zero of $\tilde{\sigma}_p(x) - \sigma_p(x) = \eps x^d$.
As $\eps\ne 0$, this implies $\alpha_i=0$ and $d>0$.

Now we show the reverse implication ``$\Leftarrow$''.
Since $\alpha_i=0$ is a zero of $\sigma_p(x)$, we have $x\mid \sigma_p(x)$,
and thus $x\mid \eps x^d+\sigma_p(x) = \tilde{\sigma}_p(x)$.
Consequently, $\alpha_i$ is a zero of $\tilde{\sigma}_p(x)$, and we get $i\in \I_{\tilde p}$.
Thus we have $i\in \I_p\cap \I_{\tilde p}$, and~(b) says that either $p=\tilde p$
or $\I_p\cap \I_{\tilde p} = \{i\}$.
In the second case, we are already done. So, assume that $p=\tilde p$.
Then we deduce that, for all $j\in \I_p = \I_{\tilde p}$, the element $\alpha_j$
is a zero of both $\sigma_p(x)$ and $\tilde{\sigma}_p(x)$.
Hence~$\alpha_j$ is a zero of $\tilde\sigma_p(x)-\sigma_p(x) = \eps x^d$ for all
$j\in \I_p = \I_{\tilde p}$.
Thus we have $\alpha_j=0$ for all $j\in \I_p$, and since the elements $\alpha_1,\dots,\alpha_n$
are pairwise distinct, we must have $\# (\I_p \cap\I_{\tilde p}) = 1$, as claimed.
\end{proof}

From a \mbox{BIG-N} fault injection one can derive polynomial equations in the unknown
support~$\alpha$, as the next remark explains.

\begin{rem}\label{rmk:genfa}
Let $(p,d,\tilde p)$ be a \mbox{BIG-N} fault injection using the fault $\eps\in\F_{2^m}$,
and assume that $\wt(\tilde{p}) \ge 2$.
Choosing $i,j\in \I_{\tilde{p}}$ with $i\ne j$, we obtain
\[
0 \;=\; \eps\, \alpha_i^d \alpha_j^d + \tprod_{k\in \I_p} \,
(\alpha_i -\alpha_k)\, \alpha_j^d \;=\; \eps\, \alpha_i^d \alpha_j^d + \tprod_{k\in \I_p} \,
(\alpha_j -\alpha_k)\, \alpha_i^d
\]
Therefore the tuple $(\alpha_1,\dots,\alpha_n)$ is a zero of the
polynomial
\[
x_i^d\, \tprod_{k\in \I_p}(x_j-x_k) \;-\; x_j^d\, \tprod_{k\in \I_p}(x_i-x_k)
\]
in $\F_{2^m}[x_1,\dots,x_n]$. Of course, using multiple \mbox{BIG-N} fault injections,
we can generate a polynomial system which has $(\alpha_1,\dots,\alpha_n)$
as one of its $\F_{2^m}$-rational solutions.

Notice that each polynomial has degree $d+\wt(p)$ and involves either
$\wt(p)+1$ or $\wt(p)+2$ indeterminates.
Moreover, from one fault injection we obtain $\binom{\wt(\tilde{p})}{2}$ polynomials,
and we have $\wt(\tilde{p})\le \deg(\tilde{\sigma}_p)\le \max(\wt(p),d)$.
Therefore we should choose both~$d$ and~$\wt(p)$ small, so that the polynomial system
contains only small degree polynomials in relatively few indeterminates.
But note that then only a few equations can be obtained from each injection, and hence
we have to perform a large number of fault injections
in order to obtain a polynomial system that involves all indeterminates.
\end{rem}

In the next section we present two specific \mbox{BIG-N} fault injection classes
which allow us to obtain equations of degree even lower than $d+\wt(p)$.
In particular, we will see how to generate linear and quadratic equations
in merely a few indeterminates.

%%%%%%%%%%%%%%%%%%%%%%%%%%%%%%%%%%%%%%%%%%%%%%%%%%%%%%%%%%%%%%%
%
% Section 5. Constant and Quadratic Fault Injection Sequences
%
%%%%%%%%%%%%%%%%%%%%%%%%%%%%%%%%%%%%%%%%%%%%%%%%%%%%%%%%%%%%%%%

\section{Constant and Quadratic Fault Injection Sequences}

In this section we construct algorithms which repeatedly perform \mbox{BIG-N} fault injections
until we obtain a linear or quadratic polynomial
satisfied by the support tuple $(\alpha_1,\dots,\alpha_n)$.
We continue to use the setting of the preceding section: let
$(K\secr,K\pub,\cC,\cP,\encr,\decr)$ be a \mbox{BIG-N} cryptosystem, and assume that the
decryption map $\decr$ satisfies Assumptions~\ref{ass:1}, \ref{ass:2}, and~\ref{ass:3}.
For $i=1,\dots,n$, let $e^{(i)}$ denote the $i$-th standard basis vector of~$\F_2^n$.
Subsequently, we are mainly interested in the following types of fault injections.

\begin{defi}\label{def:constfaultinj}
Let $(p,d,\tilde{p})$ be a \mbox{BIG-N} fault injection.
\begin{enumerate}
\item[(a)] The fault injection $(p,d,\tilde p)$ is called a {\bf constant injection}
if $d=0$ and we have $p=e^{(i_1)}+e^{(i_2)}\in\F_2^n$ for some $i_1,i_2\in\{1,\dots,n\}$
such that $i_1\ne i_2$.

\item[(d)] The fault injection $(p,d,\tilde p)$ is called a {\bf quadratic injection}
if $d=2$ and we have $p=e^{(i)}\in\F_2^n$ for some $i\in\{1,\dots,n\}$.

\item[(c)] A constant or quadratic injection $(p,d,\tilde p)$ is called {\bf successful}
if we have $\wt(\tilde{p})=2$.
\end{enumerate}
\end{defi}

The term \emph{successful} is adequately chosen, as the following proposition shows.

\begin{prop}\label{prop:faultattackeqn}
Let $(p,d,\tilde p)$ be a \mbox{BIG-N} fault injection.
\begin{enumerate}
\item[(a)] Suppose that a constant injection with $p=e^{(i_1)}+e^{(i_2)}$,
where $i_1,i_2\in\{1,\dots,n\}$ and $i_1\ne i_2$, is successful. Let
$\I_{\tilde{p}}=\{j_1,j_2\}$. Then we have
\[
\alpha_{i_1}+\alpha_{i_2}=\alpha_{j_1}+\alpha_{j_2}
\]

\item[(b)] Suppose that a quadratic fault injection with $p=e^{(i)}$
for some $i\in\{1,\dots,n\}$ is successful. Let $\I_{\tilde{p}}=\{j_1,j_2\}$.
Then we have
\[
\alpha_i\alpha_{j_1} + \alpha_i\alpha_{j_2} + \alpha_{j_1}\alpha_{j_2} \;=\; 0
\]
If, additionally, $\alpha_i\ne 0$, then we also have $\alpha_{j_1}\ne 0$ and $\alpha_{j_2}\ne 0$.
\end{enumerate}
\end{prop}

\begin{proof}
First we show~(a).
Let $\eps\in\F_{2^m}$ be the fault that is used by the fault injection $(p,d,\tilde p)$.
By definition of the error-locator polynomial, we have
$\sigma_p(x)=(x-\alpha_{i_1})(x-\alpha_{i_2})$.
Then $\tilde{\sigma}_p(x)=\sigma_p(x)+\eps$, and for $i\in\{1,\dots,n\}$
we have $\tilde{p}_i=1$ if and only if $\tilde{\sigma}_p(\alpha_i)=0$.
Using $j_1,j_2\in \I_{\tilde p}$, we get $\tilde{\sigma}_p(\alpha_{j_1}) =
\tilde\sigma_p(\alpha_{j_2})=0$. Then $\deg(\tilde{\sigma}_p)=2$ and
$j_1\ne j_2$ yield
$\tilde{\sigma}_p(x)=(x-\alpha_{j_1})(x-\alpha_{j_2})$.
Thus $\tilde{\sigma}_p(x)=\sigma_p(x)+\eps$ implies
\[
x^2+(\alpha_{i_1}+\alpha_{i_2})x+(\alpha_{i_1}\alpha_{i_2}+\eps) \;=\;
x^2+(\alpha_{j_1}+\alpha_{j_2})x+\alpha_{j_1}\alpha_{j_2}
\]
Comparing coefficients yields $\alpha_{i_1}+\alpha_{i_2}=\alpha_{j_1}+\alpha_{j_2}$,
as claimed.

Next we prove~(b). Let $i\in\{1,\dots,n\}$ and $p=e^{(i)}$.
By definition, the tuple $(p,2,\tilde p)$
is a \mbox{BIG-N} fault injection. Let $\eps\in\F_{2^m}$ be the fault used by this fault injection.
By definition of the error-locator polynomial, we have $\sigma_p(x)=x-\alpha_i$.
Hence we have $\tilde{\sigma}_p(x)=\eps x^2+x-\alpha_i$. Notice that, for $i\in\{1,\dots,n\}$,
we have $\tilde{p}_i=1$ if and only if $\tilde{\sigma}_p(\alpha_i)=0$.
Since $j_1,j_2\in \I_{\tilde p}$, we get $\tilde{\sigma}_p(\alpha_{j_1}) =
\tilde{\sigma}_p(\alpha_{j_2})=0$,
and therefore $j_1\ne j_2$ implies
$\eps x^2+x-\alpha_i=\eps (x-\alpha_{j_1})(x-\alpha_{j_2})$.
Comparing coefficients yields $1=-\eps(\alpha_{j_1}+\alpha_{j_2})$ and
$-\alpha_i=\eps\alpha_{j_1}\alpha_{j_2}$.
By multiplying the second equation with $-(\alpha_{j_1}+\alpha_{j_2})$, we get
\[
(\alpha_{j_1}+\alpha_{j_2})\cdot \alpha_i  \;=\; (-\eps(\alpha_{j_1}+\alpha_{j_2}))
\alpha_{j_1}\alpha_{j_2}  \;=\;  -\alpha_{j_1}\alpha_{j_2}
\]
and the first claim of~(b) follows.

To show the second claim, let $\alpha_i\ne 0$. We want to prove that $\alpha_{j_1}\ne 0$
and $\alpha_{j_2}\ne 0$. For a contradiction, assume that $\alpha_{j_1}=0$ or $\alpha_{j_2}=0$.
In both cases we have $\alpha_{j_1}\alpha_{j_2}=0$, and therefore the above equation
yields $-\alpha_i=\eps\alpha_{j_1}\alpha_{j_2}=0$, in contradiction to the fact that $\alpha_i\ne 0$.
Hence we have $\alpha_{j_1}\neq0$ and $\alpha_{j_2}\neq0$.
\end{proof}

Part~(a) of this proposition can now be exploited for a fault injection sequence algorithm
which finds a linear equation for $(\alpha_1,\dots,\alpha_n)$, if it terminates.

\begin{algo}{\bf (A Constant Fault Injection Sequence)}\label{alg:constfinjseq}\\
Given a \mbox{BIG-N} cryptosystem $(K\secr,K\pub,\cC,\cP,\encr,\decr)$ as above and
two distinct indices $i_1,i_2\in\{1,\dots,n\}$, consider the following instructions.
\begin{enumerate}
\item[(1)] Perform a \mbox{BIG-N} fault injection for the word $p=e^{(i_1)}+e^{(i_2)}$ in
degree zero, and let $\tilde{p}\in\F_2^n$ be its output.

\item[(2)] If $\wt(\tilde{p})=2$ and $\tilde{p}\ne p$, then write $\I_{\tilde p}=\{j_1,j_2\}$,
return the polynomial
\[
x_{i_1}+x_{i_2}+x_{j_1}+x_{j_2}\in\F_{2^m}[x_1,\dots,x_n]
\]
and stop. Otherwise, continue with~(1).
\end{enumerate}
This is a Las Vegas algorithm, i.e., it may not terminate, but
if it does terminate, then it returns a linear polynomial
$f\in \F_{2^m}[x_1,\dots,x_n]$ such that $f(\alpha_1,\dots,\alpha_n)=0$.
\end{algo}

\begin{proof}
To prove correctness, it suffices to note that, if the
algorithm stops in step~(2), Proposition~\ref{prop:faultattackeqn}.a
can be applied and yields $\alpha_{i_1}+\alpha_{i_2}=\alpha_{j_1}+\alpha_{j_2}$.
\end{proof}

Naturally, the question arises how many faults have to be injected
on average until this algorithms stops.
It turns out that in the case $n=2^m$ the precise number is given by $\tfrac{2^m}{2^{m-1}-1}\approx 2$.
For the general case $n\leq 2^m$, the probability of a successful constant fault injection
is estimated in Table~\ref{table:avgvar} for a selection of parameters (see
Section~\ref{sec:timings}).

Similarly, also Proposition~\ref{prop:faultattackeqn}.b can be
used via repeated fault injections to gain a quadratic equation satisfied
by $(\alpha_1,\dots,\alpha_n)$.
Moreover, it allows us to check whether $\alpha_i=0$ for a given $i\in\{1,\dots,n\}$.

\begin{algo}{\bf (A Quadratic Fault Injection Sequence)}\label{alg:quadfinjseq}\\
Given a \mbox{BIG-N} cryptosystem $(K\secr,K\pub,\cC,\cP,\encr,\decr)$ as above and
$i\in\{1,\dots,n\}$, consider the following sequence of instructions.
\begin{enumerate}
\item[(1)] Perform a \mbox{BIG-N} fault injection for the word $p=e^{(i)}$ in degree~2, and
let $\tilde{p}\in\F_2^n$ be its output.

\item[(2)] If $\wt(\tilde{p})>1$ and $i\in \I_{\tilde p}$ then
return $x_i\in\F_{2^m}[x_1,\dots,x_n]$  and stop.

\item[(3)] If $\wt(\tilde{p})=2$ then write $\I_{\tilde p}=\{j_1,j_2\}$,
return the polynomial
\[
x_i x_{j_1} + x_i x_{j_2} + x_{j_1} x_{j_2} \in\F_{2^m}[x_1,\dots,x_n]
\]
and stop. Otherwise, continue with~(1).
\end{enumerate}
This is a Las-Vegas algorithm. If it terminates, it returns a linear or
quadratic polynomial $f\in \F_{2^m}[x_1,\dots,x_n]$ such that $f(\alpha_1,\dots,\alpha_n)=0$.
Moreover, if the algorithm stops in step~(3) then we have
$\alpha_i,\alpha_{j_1},\alpha_{j_2} \in \F_{2^m}\setminus \{0\}$.
\end{algo}

\begin{proof}
To show correctness, we distinguish two cases.
If the algorithm terminates in step~(2), it suffices to prove $\alpha_i=0$.
Let $(e^{(i)},2,\tilde p)$ be the quadratic fault injection of step~(1), and let
$\eps\in\F_{2^m}$ be the fault that it uses.
Since $\wt(\tilde{p})>1$, we have $\wt(e^{(i)}) = 1 < \wt(\tilde{p})$,
and hence Proposition~\ref{prop:faulteq2}.a yields $\eps\ne 0$.
Now $i\in \I_{\tilde p}$ and Proposition~\ref{prop:faulteq2}.c imply $\alpha_i=0$,
as claimed. This also proves that, if the algorithm terminates
in step~(3), we must have $\alpha_i\neq0$.

Next, assume that the algorithm terminates in step~(3). We just saw that this
forces~$\alpha_i$ to be non-zero. Let $(e^{(i)},2,\tilde p)$ be the quadratic
fault injection of step~(1). Since $\wt(\tilde{p})=2$, we write $\I_{\tilde{p}}
=\{j_1,j_2\}$ and note that Proposition~\ref{prop:faultattackeqn}.b yields
$\alpha_i\alpha_{j_1}+\alpha_i\alpha_{j_2}+\alpha_{j_1}\alpha_{j_2}=0$.
Thus $(\alpha_1,\dots,\alpha_n)$ is a zero of the given polynomial~$f$.
Moreover, as $\alpha_i\neq0$, we also get $\alpha_{j_1}\ne 0$ and $\alpha_{j_2}\ne 0$
by the same proposition.
\end{proof}

Again, it is not clear how many fault injections are typically required for one
execution of the algorithm. In the case $n=2^m$, it turns out that on average
$\tfrac{2^m}{2^{m-1}-1}\approx 2$ faults  need to be injected if $\alpha_i\neq0$,
and otherwise $\tfrac{2^m}{2^m-1}\approx 1$ faults are required.
For the general setting $n\leq 2^m$, estimates for this number can be found in
Table~\ref{table:avgvar} (see Section~\ref{sec:timings}).

%%%%%%%%%%%%%%%%%%%%%%%%%%%%%%%%%%%%%%
%
%  Section 6: The BIG-N Fault Attack
%
%%%%%%%%%%%%%%%%%%%%%%%%%%%%%%%%%%%%%%

\section{The BIG-N Fault Attack}\label{sec:fatt}

In this section we first derive some simplifications of the systems of polynomial
equations we obtain by performing constant and quadratic fault injection sequences
(see Algorithms~\ref{alg:constfinjseq} and~\ref{alg:quadfinjseq}).
After that we present a strategy to determine all solutions of such a system.
Finally, we explain how one can check whether they can be extended to an
alternative secret pair.
To begin with, consider the following property of Goppa codes.

\begin{prop}\label{prop:calph}
Let $\alpha\in \F_{2^m}^n$ be a support tuple, let $g\in \F_{2^m}[x]$
be a Goppa polynomial for~$\alpha$, and
let $a\in\F_{2^m}\setminus\{0\}$. Then we have
\[
\Gamma(\alpha,g(x)) \;=\;  \Gamma(a\cdot \alpha,g(a^{-1}\cdot x))
\]
\end{prop}

\begin{proof}
For $c\in \Gamma(\alpha,g)$, we let $\hats_{c,\alpha}(x) =
\sum_{i\in \I_c}\prod_{j\in \I_c\setminus\{i\}}(x-\alpha_j)$.
By Remark~\ref{rmk:suppcodegoppapoly}, we have $g\mid \hats_{c,\alpha}$.
By applying the ring homomorphism $\Psi_a:\; \F_{2^m}[x] \longrightarrow \F_{2^m}[x]$
defined by $\Psi_a(x)=a^{-1}\cdot x$, we get $\Psi_a(g(x))\mid\Psi_a(\hats_{c,\alpha}(x))$,
where
\[
\Psi_a(\hats_{c,\alpha}(x)))  \;=\; \hats_{c,\alpha}(a^{-1}\cdot x)
\;=\;  a^{-(\wt(c)-1)}\hats_{c,a\cdot \alpha}(x)
\]
Thus $\Psi_a(g)\mid \hats_{c,a\cdot\alpha}$, and using Remark~\ref{rmk:suppcodegoppapoly}
we see that $c\in\Gamma(a\cdot \alpha,\Psi_a(g))$.
Therefore we have $\Gamma(\alpha,g)\subseteq \Gamma(a\cdot \alpha,g(a^{-1}\cdot x))$.

Conversely, we apply this inclusion to the Goppa code $\Gamma(a\alpha,g(a^{-1} x))$
with the factor $a^{-1}$ and obtain $\Gamma(a\cdot\alpha,g(a^{-1}\cdot x))  \subseteq
\Gamma(a^ {-1}\cdot a\cdot \alpha,g(a\cdot a^{-1}\cdot x))=\Gamma(\alpha,g)$.
This finishes the proof.
\end{proof}

In the following we fix a \mbox{BIG-N} cryptosystem
$(K\secr,K\pub,\cC,\cP,\encr,\decr)$ and use the notation introduced
in the preceding sections.

\begin{rem}\label{rem:fixone}
Let $L\subseteq\F_{2^m}[x_1,\dots,x_n]$ be a set of polynomials
obtained by constant and quadratic fault injection sequences.
By
\[
\cZ_{\F_{2^m}}(L) = \{(a_1,\dots,a_n)\in\F_{2^m}^n \mid f(a_1,\dots,a_n)=0
\hbox{\ for all\ } f\in L\}
\]
we denote the zero set of~$L$ in~$\F_{2^m}$. Now we consider the set
\[
S_L \;=\;\{(a_1,\dots,a_n)\in\cZ_{\F_{2^m}}(L)\mid a_i\ne a_j\text{\ for\ }i\ne j\}
\]
Since every $f\in L$ results from
Algorithms~\ref{alg:constfinjseq} or~\ref{alg:quadfinjseq}, we have
$f(\alpha_1,\dots,\alpha_n)=0$, and since~$\alpha =(\alpha_1,\dots,\alpha_n)$
is a support tuple, we even get $\alpha\in S_L$.

Notice that the polynomials in~$L$ are homogeneous. Hence the set~$L$ generates
a homogeneous ideal in $\F_{2^m}[x_1,\dots,x_n]$. Consequently,
for $(a_1,\dots,a_n)\in S_L$ and $b\in\F_{2^m}\setminus\{0\}$ we also have
$(ba_1,\dots,ba_n)\in S_L$. In view of Proposition~\ref{prop:calph} and the fact
that $g(x)$ is irreducible if and only if $g(b^{-1}\cdot x)$ is irreducible,
we can therefore assume without loss of generality
that one non-zero support element $\alpha_i$ with $i\in\{1,\dots,n\}$ is chosen
arbitrarily to be~1.

This means that we may choose an index $i\in\{1,\dots,n\}$ for which we know
that $\alpha_i\ne 0$. Then we consider the dehomogenization of~$L$ with respect
to the indeterminate~$x_i$, i.e., we add the polynomial $x_i-1$ to~$L$.
To determine such an index we can use a quadratic fault injection sequence.
Clearly, it is best to choose the indeterminate that occurs most frequently in the
polynomials of~$L$, in order to simplify the fault equations system as much as possible.
\end{rem}

In view of this remark, we are led to the following definitions.

\begin{defi}\label{def:feqs}
Let $C=\Gamma(\alpha,g)$ be a binary irreducible Goppa code, and let
$(K\secr,K\pub,\cP,\allowbreak \cC,\encr,\decr)$ be a \mbox{BIG-N} cryptosystem which uses~$C$.
\begin{enumerate}
\item[(a)] Let $L_1\subseteq\F_{2^m}[x_1,\dots,x_n]$ be a set of linear polynomials
obtained from constant fault injection sequences (see Algorithm~\ref{alg:constfinjseq}).
Let $L_2\subseteq\F_{2^m}[x_1,\dots,x_n]$ be a set of linear and quadratic polynomials obtained
from quadratic fault injection sequences (see Algorithm~\ref{alg:quadfinjseq}).
Let $i\in\{1,\dots,n\}$ be such that the indeterminate~$x_i$ occurs in some quadratic
polynomial of~$L_2$. Then the set $L_1\cup L_2\cup\{x_i-1\}$
is called a {\bf fault equation system}.

\item[(b)] Given a set of polynomials $L\subseteq\F_{2^m}[x_1,\dots,x_n]$, we call
\[
S_L  \;=\; \{(a_1,\dots,a_n)\in\cZ_{\F_{2^m}}(L)
           \mid a_i\ne a_j\text{\ for\ } i\neq j\}  \;\subseteq\;\F_{2^m}^n
\]
the {\bf support candidate set} of~$L$, and every element $(a_1,\dots,a_n)\in S_L$ is called
a {\bf support candidate} of~$L$.
\end{enumerate}
\end{defi}

Given a fault equation system~$L$, Remark~\ref{rem:fixone}
shows that we may assume $\alpha\in S_L$. Hence the task of finding~$\alpha$
is reduced to solving a suitable fault equation system.
In view of Algorithm~\ref{alg:altkey}, we may reduce this task further.

\begin{defi}
Let $C=\Gamma(\alpha,g)$ be a binary Goppa code, let
$\tilde\alpha\in\F_{2^m}^n$ be a support tuple, and let $u\in\N_+$.
\begin{enumerate}
\item[(a)]  A Goppa polynomial $\tilde g\in\F_{2^m}[x]$ for $\tilde\alpha$ with
$C\subseteq\Gamma(\tilde\alpha,\tilde g)$ and $\deg(\tilde g) = u$ is called a
{\bf degree-$u$ extension} of $\tilde\alpha$ with respect to~$C$.

\item[(b)] A support tuple $\tilde\alpha$ is called {\bf degree-$u$ extendable}
to~$C$ if there exists a degree-$u$ extension of~$\tilde\alpha$ with respect to~$C$.
\end{enumerate}
If the code~$C$ is clear from the context, it may also be omitted.
\end{defi}

\begin{rem}
As explained at the end of Section~\ref{sec:bignc}, one can
apply Algorithm~\ref{alg:altkey} to break a \mbox{BIG-N} cryptosystem
which uses a binary Goppa code $C=\Gamma(\alpha,g)$ with $\deg(g)=t$
as follows: Find a support tuple $\tilde\alpha\in\F_{2^m}^n$
and a Goppa polynomial $\tilde g\in\F_{2^m}[x]$ with
$\deg(\tilde g)=u\ge 2t$ such that $\tilde{g}$ is a degree-$u$ extension
of~$\tilde\alpha$ with respect to~$C$. This is the reason why $(\tilde{\alpha},
\tilde{g})$ is called an alternative secret pair.

In order to find such a support tuple $\tilde\alpha\in\F_{2^m}^n$,
we generate a fault equation system $L\subseteq\F_{2^m}[x_1,\dots,x_n]$
using constant and quadratic fault injection sequences.
Note that $C=\Gamma(\alpha,g^2)$ and $\alpha\in S_L$ imply that
at least one degree-$2t$ extendable support candidate can be found in~$S_L$.
\end{rem}

Based on these observations, we now perform a \mbox{BIG-N} fault attack which
is based on the following three steps:
Firstly, we compute the support candidate set~$S_L$ of a fault equation system~$L$.
Secondly, we determine a degree-$u$ extendable support candidate $\tilde\alpha\in S_L$
with $u\ge 2t$ and its corresponding extension. Finally, we combine everything and
compute an alternative secret pair, thereby breaking the system.

%%%%%%%%%%%%%%%%%%%%%%%%%%%%%%%%%%%%%%%%%%%%%
%
% Subsection 6.1. Finding Support Candidates
%
%%%%%%%%%%%%%%%%%%%%%%%%%%%%%%%%%%%%%%%%%%%%%

\subsection{Finding Support Candidates}\label{sec:findingsuppcand}

Let $L\subseteq\F_{2^m}[x_1,\dots,x_n]$ be a fault equation system.
Then the problem of computing the support candidate set~$S_L$
can be reduced to finding the zero set $\cZ_{\F_{2^m}}(L)$.
Computing this zero set is a classic problem of computer algebra
which can be solved, for instance, using Gr\"obner basis techniques
(see \cite[Sec. 3.7]{kreuzer2008comp} or~\cite[Sec. 6.3]{kreuzer2016comp}).
However, in order to improve the efficiency of these methods,
it is important that we first use the linear equations in~$L$
to eliminate some indeterminates and reduce the complexity of the quadratic equations.
The following algorithm aids this task.

\begin{algo}{\bf (Solving a Fault Equation System)}\label{alg:solve}\\
Let $\mathcal{N}$ be a \mbox{BIG-N} cryptosystem, and let $L\subseteq\F_{2^m}[x_1,\dots,x_n]$
be a fault equation system obtained by applying constant and quadratic fault
injection sequences to~$\mathcal{N}$. Moreover, let $\sigma$ be a term ordering.
Consider the following sequence of instructions.
\begin{enumerate}
\item[(1)] Let $L\lin=\{f\in L\mid \deg(f)=1\}$.

\item[(2)] Interreduce the set $L\lin$ linearly and get a set
$L^{\rm irlin}=\{\ell_1,\dots,\ell_r\}$ such that, for $i=1,\dots,r$, the indeterminate
$\LT_\sigma(\ell_i)$ does not occur in the other polynomials of $L^{\rm irlin}$.
Renumber the indeterminates such that $\LT_\sigma(\ell_i)=x_i$ for $i\in\{1,\dots,r\}$.

\item[(3)] Define a ring homomorphism $\Psi:\; \F_{2^m}[x_1,\dots,x_n]
\longrightarrow \F_{2^m}[x_{r+1},\dots,x_n]$
by $\Psi(x_i)=\ell_i+x_i$ for $i\in\{1,\dots,r\}$, and by
$\Psi(x_i)=x_i$ for $i\in\{r+1,\dots,n\}$.

\item[(4)] Let $L\red=\Psi(L) \setminus \{0\}$.
Compute $S\red = \cZ_{\F_{2^m}}(L\red)$ using Gr\"obner basis techniques.

\item[(5)] Define a map $\psi:\; \F_{2^m}^{n-r} \longrightarrow \F_{2^m}^n$ by
$\psi(\gamma)=(\Psi(x_1)(\gamma),\dots,\Psi(x_n)(\gamma))$.

\item[(6)] Return $S=\{(\tilde\alpha_1,\dots,\tilde\alpha_n)\in \psi(S\red)\mid
\tilde\alpha_i\neq\tilde\alpha_j\text{ for } i\neq j\}$.
\end{enumerate}
This is an algorithm which computes the support candidate set~$S_L$ of~$L$.
\end{algo}

\begin{proof}
Since finiteness is clear, we have to show that the set~$S$
returned in step~(6) is indeed equal to~$S_L$.

First we show the inclusion $S\subseteq S_L$.
Let $\gamma\in S\red$ with $\psi(\gamma)\in S$.
By construction, we have $\Psi(f)\in L\red$ for all $f\in L$.
Hence we get $0=\Psi(f)(\gamma) = f(\Psi(x_1)(\gamma),\dots,\Psi(x_n)(\gamma))
=f(\psi(\gamma))$.
Since this equality holds for all $f\in L$, we conclude that
$\psi(\gamma)\in\cZ_{\F_{2^m}}(L)$. Therefore we have
$\psi(S\red)\subseteq\cZ_{\F_{2^m}}(L)$, and by the construction of~$S$ also $S\subseteq S_L$.

Conversely, let $\tilde\alpha=(\tilde\alpha_1,\dots,\tilde\alpha_n)\in S_L$.
To prove $\alpha\in S$, it suffices to show that $\tilde\alpha=\psi(\gamma)$
and $\gamma\in S\red$ for $\gamma=(\tilde\alpha_{r+1},\dots,\tilde\alpha_n)$.
By the definition of~$\Psi$ and the fact that $\ell_i(\tilde\alpha)=0$ for
$i\in\{1,\dots,r\}$, we have $\Psi(x_i)(\gamma)=\tilde\alpha_i$ for $i\in\{1,\dots,n\}$.
This yields $\psi(\gamma)=(\Psi(x_1)(\gamma),\dots,\Psi(x_n)(\gamma))=
(\tilde\alpha_1,\dots,\tilde\alpha_n)=\tilde\alpha$.

Thus it remains to show that $\gamma\in S\red=\cZ_{\F_{2^m}}(L\red)$.
Note that, by construction, we have $L\red=\{\Psi(f)\mid f\in L\}$.
Using $\Psi(x_i)(\gamma)=\tilde\alpha_i$ for $i\in\{1,\dots,n\}$
and $\tilde\alpha\in S_L$, we have
\[
\Psi(f)(\gamma)=f(\Psi(x_1)(\gamma),\dots,\Psi(x_n)(\gamma)) \;=\;
f(\tilde\alpha_1,\dots,\tilde\alpha_n) \;=\; f(\tilde\alpha) \;=\; 0
\]
for all $\Psi(f)\in L\red$. This proves $\gamma\in \cZ_{\F_{2^m}}(L\red)$,
and hence $\tilde\alpha\in S$.
\end{proof}

%%%%%%%%%%%%%%%%%%%%%%%%%%%%%%%%%%%%%%%%%%%%%%%%%%%%%%
%
% Subsection 6.2. Finding Suitable Goppa Polynomials
%
%%%%%%%%%%%%%%%%%%%%%%%%%%%%%%%%%%%%%%%%%%%%%%%%%%%%%%

\subsection{Finding Suitable Goppa Polynomials}

Recall that the generating pair $(\alpha,g)$ of the irreducible Goppa code $C=\Gamma(\alpha,g)$
is part of the secret key of any \mbox{BIG-N} cryptosystem using~$C$. The code~$C$ itself is
published via the public parity check matrix $H\pub$.
Let $L\subseteq\F_{2^m}[x_1,\dots,x_n]$ be a fault equation system, let $S_L$ be its support
candidate set, let $\tilde{\alpha}\in S_L$ be a support candidate which is degree-$u$
extendable with respect to~$C$ for some $u\ge 2t$, and let $\tilde g\in\F_{2^m}[x]$ be a degree-$u$
extension of~$\tilde{\alpha}$. This means that $C\subseteq\Gamma(\tilde\alpha,\tilde g)$.
By Remark~\ref{rmk:suppcodegoppapoly}, it follows for all $c\in \Gamma(\tilde\alpha,\tilde g)$
that we have
\[
  \tilde{g} \;\mid\; \tsum_{i\in \I_c} \tprod_{j\in \I_c\setminus\{i\}} \, (x-\tilde\alpha_j)
\]
in $\F_{2^m}[x]$.
In particular, this divisibility then holds for all $c\in C$.
Hence, knowing only~$C$ and~$\tilde\alpha$, we can compute multiples
of the desired polynomial~$\tilde g$.
The following algorithm uses these observations and computes an extension~$\tilde g$
of a given support tuple~$\tilde\alpha$ if and only if $\tilde\alpha$ is extendable.
Note that its core idea is based on~\cite[p. 125]{overbeck2009code}.

\begin{algo}[{\tt GoppaGCD}]\label{alg:GoppaGCD}\hfill

{\bf Input:} A support tuple $\tilde\alpha\in\F_{2^m}^n$, $t\in\N_+$,
and a binary Goppa code $C\subseteq\F_2^n$.

{\bf Output:} {\tt Fail}, or a degree-$2t$ extension of $\tilde\alpha$
with respect to~$C$.

\begin{algorithmic}[1]
  \STATE Let $\tilde g:=0$, and let $B$ be an $\F_2$-basis of~$C$.
  \FOR{$c\in B$}
    \STATE $\hats_{c,\tilde\alpha} := \sum_{i\in \I_c}\prod_{j\in \I_c\setminus\{i\}}
           (x-\tilde\alpha_j)$
    \STATE $\tilde g:=\gcd(\tilde g,\hats_{c,\tilde\alpha})$
    \IF{$\deg(\tilde g)<2t$}
       \RETURN{{\tt Fail}}
    \ENDIF
  \ENDFOR
  \FOR{$i=1,\dots,n$}
    \WHILE{$\tilde g(\tilde \alpha_i)=0$}
      \STATE $\tilde g:=\tfrac{\tilde g}{x-\tilde\alpha_i}$
    \ENDWHILE
  \ENDFOR
  \IF{$\deg(\tilde g)<2t$}
     \RETURN{{\tt Fail}}
  \ELSE
    \RETURN{$\tilde{g}$}
  \ENDIF
\end{algorithmic}
\end{algo}

\begin{prop}\label{prop:goppagcd}
Let $\tilde\alpha\in\F_{2^m}^n$ be a support tuple, let $t\in\N_+$,
and let $C\subseteq\F_2^n$ be a binary Goppa code.
Then Algorithm~\ref{alg:GoppaGCD} is finite and the following conditions
are equivalent.
\begin{enumerate}
\item[(a)] The function $\GoppaGCD(\tilde\alpha,t,C)$ returns a degree-$u$
extension of $\tilde\alpha$ with respect to~$C$ for some $u\ge 2t$.

\item[(b)] The tuple $\tilde\alpha$ is degree-$u$ extendable
with respect to~$C$ for some $u\ge 2t$.
\end{enumerate}
Otherwise, the algorithm returns \textnormal{\tt Fail}.
\end{prop}

\begin{proof}
Since finiteness is clear, it suffices to prove the equivalence of~(a)
and~(b) and the additional claim. To begin with, we show that the output
of the function $\GoppaGCD(\tilde\alpha,t,C)$ is either {\tt Fail} or a
degree-$u$ extension of~$\tilde{\alpha}$ for some $u\ge 2t$.

Clearly, the algorithm ends either in step~(6), (15), or~(17). Hence it terminates
with {\tt Fail} or a polynomial $\tilde g\in\F_{2^m}[x]$.
Suppose that it ends in step~(17) with a polynomial $\tilde g$.
Notice that this entails $\deg(\tilde g)\ge 2t$, as otherwise
the algorithm would have terminated at the latest in step~(15).

Thus it remains to prove $C\subseteq\Gamma(\tilde\alpha,\tilde g)$.
By the loop in steps (9)-(13), we have $\tilde g(\tilde\alpha_i)\neq 0$
for all $i\in\{1,\dots,n\}$.
Therefore~$\tilde g$ is a Goppa polynomial for the support
tuple $\tilde\alpha$ and $\Gamma(\tilde\alpha,\tilde g)$ is a binary Goppa code.
Let $B=\{c_1,\dots,c_k\}$ be the $\F_2$-basis of~$C$ chosen in step~(1).
By the construction of~$\tilde g$ in steps~(2)-(8), we have
$\tilde g\mid \hats_{c_i,\tilde\alpha}$ for all $i\in\{1,\dots,k\}$.
Since in steps (9)-(13) only factors of~$\tilde g$ are removed, these divisibilities
still hold true in step~(17).
Therefore the polynomial $\tilde g$ returned in step~(17) satisfies
$\tilde g\mid \hats_{c_i,\tilde\alpha}$ for all $i\in\{1,\dots,k\}$.
By  Remark~\ref{rmk:suppcodegoppapoly}, this implies that
$c_i\in\Gamma(\tilde\alpha,\tilde g)$ for all $i\in\{1,\dots,k\}$.
Using the fact that~$B$ is an $\F_2$-basis of~$C$, we deduce
$C\subseteq\Gamma(\tilde\alpha,\tilde g)$.

Since the implication (a)$\Rightarrow$(b) is trivially true,
it remains to prove that~(b) implies~(a).
Let $\hat g\in\F_{2^m}[x]$ be a degree-$u$ extension of $\tilde\alpha$,
where $u\ge 2t$.
Let $B=\{c_1,\dots,c_k\}$ be the $\F_2$-basis of~$C$ chosen is step~(1).
From the hypothesis and Remark~\ref{rmk:suppcodegoppapoly}, we get that
$\hat g\mid \hats_{c_i,\tilde\alpha}$ for $i\in\{1,\dots,k\}$.
Since we compute~$\tilde g$ in steps (2)-(8) as the greatest common divisor of
the polynomials $\hats_{c_i,\tilde\alpha}$ for $i\in\{1,\dots,k\}$, we
have $\hat g\mid \tilde g$ in every iteration of the loop.
In particular, this implies that we have $\deg(\tilde g)\ge \deg(\hat g)\ge 2t$
and thus the algorithm does not terminate in this loop.
Now, in steps (9)-(13), all linear factors of the form $(x-\tilde\alpha_i)$ for $i\in\{1,\dots,n\}$
are removed, and thus we get $\tilde g(\tilde\alpha_i)\ne 0$ for all $i\in\{1,\dots,n\}$.
As $\hat g$ is a Goppa polynomial for~$\tilde\alpha$, we also have $\hat g(\tilde\alpha_i)\ne 0$
for $i\in\{1,\dots,n\}$.
Consequently, we have $\hat g\mid \tilde g$, and therefore $\deg(\tilde g)\ge \deg(\hat g)\ge 2t$.
Finally, the algorithm terminates in step~(17) by returning the polynomial $\tilde g\in\F_{2^m}[x]$.
This finishes the proof.
\end{proof}

Note that an $\F_2$-basis $B$ of $C$, as required in step~(1), can be deduced
from the rows of a generator matrix of~$C$.
Moreover, when the degree of the polynomial $\tilde g$ is equal to~$2t$ in the course
of the loop (2)-(8), the polynomial will either stay the same for the remaining
execution of the algorithm, or the result is {\tt Fail}.
Therefore we can also stop the algorithm at this point and use any other method to check
if~$\tilde g$ is a Goppa polynomial for~$\tilde\alpha$
and if we have $C\subseteq \Gamma(\tilde\alpha,\tilde g)$.
For instance, an efficient method is to compare a parity check matrix
of $\Gamma(\tilde\alpha,\tilde g)$ with $H\pub$.

%%%%%%%%%%%%%%%%%%%%%%%%%%%%%%%%%%%%%%%%%%%%%%%%%%%%%
%
% Subsection 6.3. Computing an Alternative Secret Pair
%
%%%%%%%%%%%%%%%%%%%%%%%%%%%%%%%%%%%%%%%%%%%%%%%%%%%%%

\subsection{Computing an Alternative Secret Pair}

In this subsection we combine the algorithms of the previous two
subsections and obtain the following fault attack algorithm which
returns an alternative secret pair.

\begin{algo}{\bf (\mbox{BIG-N} Fault Attack)}\label{alg:attack}\\
Let $m,t,n\in\N_+$ be such that $mt<n\leq 2^m$, let
$C=\Gamma(\alpha,g)$ be a binary irreducible Goppa code with parameters
$(m,t,n)$, and let~$\mathcal{N}$ be a \mbox{BIG-N} cryptosystem which uses~$C$.
Suppose that the implementation of the decryption map satisfies Assumptions~\ref{ass:1},
\ref{ass:2}, and~\ref{ass:3}.

Let $L\subseteq\F_{2^m}[x_1,\dots,x_n]$ be a fault equation system
which has been constructed using constant and quadratic fault injection sequences
applied to~$\mathcal{N}$.
\begin{enumerate}
\item[(1)] Compute the support candidate set $S_L$ of~$L$
using Algorithm~\ref{alg:solve}.

\item[(2)] Choose $\tilde\alpha\in S_L$ and remove it from $S_L$.

\item[(3)] Use Algorithm~\ref{alg:GoppaGCD} to compute $\GoppaGCD(\tilde\alpha,t,C)$.
If the output is {\tt Fail}, then go to step~(2). Otherwise, the output is a
polynomial $\tilde g$. Return the pair $(\tilde\alpha,\tilde g)$ and stop.
\end{enumerate}
This is an algorithm which computes an alternative secret pair $(\tilde\alpha,\tilde g)$.
\end{algo}

\begin{proof}
First we prove finiteness. For the support candidate set~$S_L$ of step~(1),
we may assume $\alpha\in S_L$ by Remark~\ref{rem:fixone}.
Since $S_L\subseteq\F_{2^m}^n$ is a finite set, after finitely many iterations
of steps~(2) and~(3) the support tuple $\tilde\alpha=\alpha\in S_L$ is chosen.
Since $\alpha$ is degree-$2t$ extendable using the extension~$g^2$,
Proposition~\ref{prop:goppagcd} shows that the call to $\GoppaGCD(\tilde\alpha,t,C)$
in step~(3) returns a polynomial $\tilde g\in\F_{2^m}[x]$. Hence the algorithm terminates.

It remains to prove that the output $(\tilde\alpha,\tilde g)$ is an alternative secret pair.
By Proposition~\ref{prop:goppagcd}, the algorithm terminates in step~(3)
if and only if $\tilde\alpha\in S_L$ is degree-$u$ extendable for some $u\ge 2t$.
In this case the polynomial $\tilde g$ is a degree-$u$ extension of~$\tilde\alpha$ for
some $u\geq 2t$, and Algorithm~\ref{alg:altkey} implies that $(\tilde\alpha,\tilde g)$
is an alternative secret pair.
\end{proof}

%%%%%%%%%%%%%%%%%%%%%%%%%%%%%%%%%%%%%%
%
%  Section 7: Experiments and Timings
%
%%%%%%%%%%%%%%%%%%%%%%%%%%%%%%%%%%%%%%

\section{Experiments and Timings}\label{sec:timings}

In this section we apply the \mbox{BIG-N} fault attack (Algorithm~\ref{alg:attack}) to a selection
of state-of-the-art security levels.
Table~\ref{table:secpar} contains a list of recommended parameter choices for the Goppa codes
to be used in the \mbox{BIG-N} cryptosystem along with their claimed security.

\begin{table}[!ht]
  \centering
  \begin{tabular}{lrrr}
    \toprule
    \textbf{Security Level} & \multicolumn{1}{c}{$n$} & \multicolumn{1}{c}{$m$} & \multicolumn{1}{c}{$t$} \\ \midrule
    insec\hfill (60bit) & 1024 & 10 & 38 \\ %\midrule
    short I\hfill (80bit) & 2048 & 11 & 27\\
    short II\hfill (80bit) & 1632 & 11 & 33\\
    mid I\hfill (128bit) & 2960 & 12 & 56\\
    mid II\hfill (147bit) & 3408 & 12 & 67\\ %\midrule
    long I\hfill (191bit) & 4624 & 13 & 95\\
    long II\hfill (256bit) & 6624 & 13 & 115\\
    long III\hfill (266bit) & 6960 & 13 & 119\\
    \bottomrule\\
\end{tabular}
\caption{Security parameters for BIG-N cryptosystems proposed
in~\cite{bernstein2008attacking}.\label{table:secpar}}
\end{table}

Recall that the only input of the \mbox{BIG-N} fault attack is a fault equation system.
It can be computed, for instance, using the following algorithm.
Our experiments show that fault equation systems obtained with this method
can be solved efficiently with Algorithm~\ref{alg:solve}.

\begin{algo}{\bf (Computing a Fault Equation System)}\label{alg:feq}\\
Let $m,t,n\in\N+$ be such that $mt<n\leq 2^m$, let
$C=\Gamma(\alpha,g)$ be a binary irreducible Goppa code with parameters
$(m,t,n)$, and let~$\mathcal{N}$ be a \mbox{BIG-N} cryptosystem which uses~$C$.
Suppose that the implementation of the decryption map satisfies Assumptions~\ref{ass:1},
\ref{ass:2}, and~\ref{ass:3}.
Consider the following sequence of instructions.
\begin{enumerate}
\item[(1)] Let $I:=\{(n,1)\}\cup\{(i,i+1) \mid i\in\{1,\dots,n-1\}\}\subseteq\{1,\dots,n\}^2$.

\item[(2)] For each pair $(i_1,i_2)\in I$, perform a constant fault injection
sequence and collect the resulting linear polynomials
in $L_1\subseteq\F_{2^m}[x_1,\dots,x_n]$.

\item[(3)] Let $I\subseteq\{1,\dots,n\}$ with
$\# I=\lfloor \tfrac{n}{10}\rfloor$ be chosen uniformly at random.

\item[(4)] For each index $i\in I$, perform a quadratic fault injection sequence
and collect the resulting polynomials in $L_2\subseteq\F_{2^m}[x_1,\dots,x_n]$.

\item[(5)] Let $x_i$ be an indeterminate which occurs most often in the terms
of the polynomials of~$L_1$ and which appears in $L_2\setminus\{x_1,\dots,x_n\}$.

\item[(6)] Return $L=L_1\cup L_2\cup\{x_i-1\}$ and stop.
\end{enumerate}
This is a Las-Vegas algorithm which performs $n$ constant fault injection sequences
and $\lfloor\tfrac{n}{10}\rfloor$ quadratic fault injection sequences.
It returns a fault equation system $L\subseteq\F_{2^m}[x_1,\dots,x_n]$.
\end{algo}

In order to determine the expected number of \mbox{BIG-N} fault injections needed
during the computation of a fault equation system using this algorithm,
we denote the probability of a successful constant fault injection by~$p_0$
and the probability of a successful quadratic fault injection by~$p_2$.
Then we expect that $\tfrac{n}{p_0}+\lfloor\tfrac{n}{10}\rfloor\tfrac{1}{p_2}$
faults have to be injected in total to generate the fault equation system using
Algorithm~\ref{alg:feq}.

Since there is no obvious formula to calculate these probabilities,
they were estimated in the following way:
For each security level, we considered three random irreducible Goppa codes.
For each Goppa code, we chose uniformly at random $200$ distinct words~$p$
for the fault injection sequences in step~(2) and~(4), respectively.
For each word~$p$, we computed the exact number of
faults which lead to a successful fault injection. Dividing the average of those
three numbers by~$2^m$ yields our estimates~$\hat{p}_0$ and~$\hat{p}_2$
respectively.

Table~\ref{table:avgvar} contains, both for the constant and the quadratic fault injection
sequences, the average number of faults which lead to a successful fault injection
and the estimated standard deviation. The estimates $\hat{p}_0$
and~$\hat{p}_2$ for~$p_0$ and~$p_2$ are also given.

\begin{table}[!ht]
\centering
\begin{tabular}{lcccrclrcr}
\toprule
{\bf Sec Lvl\;\;} & \multicolumn{3}{c}{\;\bf Parameters\;\;} & \multicolumn{2}{c}{\;\bf succ const inj\;\;}
    & {\;\;$\hat p_0$\;\;} & \multicolumn{2}{c}{\;\;\bf succ quad inj\;\;} & {$\hat p_2$} \\
    & $n$ & $m$ & $t$ & {avg} & {std dev} & {(\%)} & {avg} & {std dev} & {(\%)}\\
\midrule
    insec & 1024 & 10 & 38 & 511.0 & 0.0 & 49.9 & 511.0 & 0.0 & 49.9\\ %\midrule
    short I & 2048 & 11 & 27 & 1023.0 & 0.0 & 50.0 & 1023.0 & 0.0 & 50.0\\
    short II & 1632 & 11 & 33 & 649.5 & 5.3 & 31.7 & 649.5 & 5.2 & 31.7\\
    mid I & 2960 & 12 & 56 & 1067.9 & 9.3 & 26.1 & 1069.2 & 8.8 & 26.1\\
    mid II & 3408 & 12 & 67 & 1416.4 & 6.6 & 34.6 & 1417.0 & 6.5 & 34.6\\ %\midrule
    long I & 4624 & 13 & 95 & 1304.5 & 15.4 & 15.9 & 1303.7 & 15.3 & 15.9\\
    long II & 6624 & 13 & 115 & 2677.5 & 10.0 & 32.7 & 2676.8 & 9.9 & 32.7\\
    long III & 6960 & 13 & 119 & 2955.5 & 8.3 & 36.1 & 2955.8 & 8.8 & 36.2\\
\bottomrule\\
\end{tabular}
\caption{Average numbers of faults and success probabilities.\label{table:avgvar}}
\end{table}
Observe that the success probabilities drop significantly when
the ratio $\tfrac{n}{2^m}$ gets smaller.
This can be attributed to the fact that a constant or quadratic fault injection
is successful if and only if the \emph{faulty} error-locator polynomial $\tilde\sigma_p(x)$
has two zeros among the support elements $\{\alpha_1,\dots,\alpha_n\}$,
and it seems natural that there are more such faults when~$n$ is larger.

An implementation of a \mbox{BIG-N} cryptosystem following the FPGA-based implementation
of~\cite{wang2017fpga,wang2018fpga} is provided in the computer algebra system
CoCoA-5~\cite{CoCoA} along with all algorithms of this paper.
For the computation of the zero set of the \emph{reduced} fault equation system $L\red$
in step~(4) of Algorithm~\ref{alg:solve}, we make use the CoCoA-5 function
{\tt RationalSolve} which uses Gr\"obner basis computations.

Timings of the \mbox{BIG-N} fault attack (Algorithm~\ref{alg:attack}) for all security parameters
of Table~\ref{table:secpar} are given in Table~\ref{table:timings}.
We also list the size of the reduced fault equation system $L\red$,
as computed in step~(4) of Algorithm~\ref{alg:solve}, the time for computing its zero
set $\cZ_{\F_{2^m}}(L\red)$ using the CoCoA-5 function \textnormal{\tt RationalSolve},
and the time for extending a support candidate to an alternative secret pair using
Algorithm~\ref{alg:GoppaGCD}. The timings represent the average
of three runs of the algorithm applied to distinct randomly generated \mbox{BIG-N}
cryptosystems.

\begin{table}[!ht]
\centering
\begin{tabular}{lccccccr}
\toprule
\multirow{2}{*}{\bf sec lvl} & \multicolumn{3}{c}{\bf interreduced $L\red$} & {\bf RatSol}
    & {\bf Alg~\ref{alg:GoppaGCD}} & {\bf total} & {\bf exp no req}\\
    & ind & {eq} & {(s)} & {(s)} & {(s)} & {(s)} & {\bf  fault inj}\\
\midrule
    insec    & 10    & 46    & 4.5   & 2.6  & 13.1  & 20.5   &  2258.41\\
    short I  & 10-11 & 55-56 & 28.4  & 4.1  & 19.0  & 52.9   &  4510.40\\
    short II & 11    & 56    & 15.7  & 4.2  & 19.0  & 39.8   &  5563.51\\
    mid I    & 12    & 67    & 88.0  & 7.1  & 71.2  & 169.5  & 12474.52\\
    mid II   & 11-12 & 66-67 & 119.9 & 6.3  & 101.7 & 231.9  & 10840.84\\
    long I   & 13    & 79    & 304.2 & 13.1 & 240.9 & 566.1  & 31946.16\\
    long II  & 13    & 79    & 866.1 & 13.1 & 438.2 & 1333.8 & 22295.50\\
    long III & 13    & 79    & 938.4 & 13.2 & 481.0 & 1450.9 & 21220.49\\
\bottomrule\\
\end{tabular}
\caption{Timings of the BIG-N fault attack (Algorithm~\ref{alg:attack}).\label{table:timings}}
\end{table}
This table shows that a straightforward implementation of the \mbox{BIG-N} cryptosystem
using classical decoding methods is quite susceptible to the \mbox{BIG-N} fault attack.
Even state-of-the-art security parameters were broken in about 20 minutes.
Therefore we recommended to implement the countermeasures proposed in
Remark~\ref{rmk:countermeasures}.

\subsection*{Acknowledgements}
This research was supported by DFG (German
Research Foundation) project ``Algebraische Fehlerangriffe''
grant {KR 1907/6-2}.

%%%%%%%%%%%%%%%%%%%%%%%%%%%%%%%%%%%%%%
%   Bibliography
%%%%%%%%%%%%%%%%%%%%%%%%%%%%%%%%%%%%%%

\bibliography{BIG-N-Fault-Attack}
\bibliographystyle{plain}

\end{document}